\RequirePackage{etex}
\PassOptionsToPackage{usenames,dvipsnames,svgnames,table}{xcolor}
\PassOptionsToPackage{hyphens}{url}
\documentclass[runningheads]{llncs}

\usepackage{latex/packages}
\theoremstyle{remark}

\AtBeginDocument{%

}

\DeclareMathOperator{\In}{In}
\DeclareMathOperator{\Id}{Id}
\DeclareMathOperator{\Out}{Out}
\DeclareMathOperator{\Occ}{Occ}

\newcommand{\prc}{\lstinline[language=C]}

\ebproofnewstyle{small}{
	separation = 1em, rule margin = .5ex,
}

\newcommand*{\eg}{e.g.\@,\xspace}

\newcommand*{\ie}{i.e.\@,\xspace}

\newcommand*{\resp}{resp.\@\xspace}

\makeatletter
\newcommand*{\etc}{%
	\@ifnextchar{.}%
	{etc}%
	{etc.\@\xspace}%
}
\makeatother

\makeatletter
\renewcommand{\boxed}[1]{\text{\fboxsep=.2em\fbox{\m@th$\displaystyle#1$}}}
\makeatother

\usepackage{xspace}

\newcommand{\sccs}{{\scshape scc}s\xspace}

\newcommand{\comm}[1]{\mathtt{#1}}
\newcommand{\DFG}{{\sc dfg}\xspace}
\newcommand{\DFGs}{{\sc dfg}s\xspace}
\newcommand{\dfg}[1]{\mathbb{M}(\comm{#1})}

\newcommand{\corr}[1]{\mathrm{Corr}(#1)}

\algtext*{EndWhile}%
\algtext*{EndIf}%

\algloopdefx{Use}[1]{\textbf{use}(#1)} %
\renewcommand{\gets}{=}

\definecolor{clem}{HTML}{60656F} %
\definecolor{thos}{HTML}{ff0066} %
\definecolor{thor}{HTML}{EB6424} %
\definecolor{neea}{HTML}{823038} %

\definecolor{darkgray}{rgb}{.4,.4,.4}

\newcommand{\replabel}{\label} %

\ExplSyntaxOn

\NewDocumentCommand{\repeattheorem}{m}
{
	\group_begin:
	\renewcommand{\thetheorem}{\ref{#1}}
	\renewcommand{\replabel}[1]{\tag{\ref{##1}}}
	\prop_item:Nn \g_reptheorem_prop { #1 }
	\endtheorem
	\group_end:
}

\NewDocumentEnvironment{reptheorem}{m+b}
{
	\prop_gput:Nnn \g_reptheorem_prop { #1 } { \theorem #2 \endtheorem }
	\theorem#2\unskip\label{#1}\endtheorem
}{}

\prop_new:N \g_reptheorem_prop

\ExplSyntaxOff
\allowdisplaybreaks[1] %

\NewDocumentCommand{\bywhom}{m}{%
	{\nobreak\hfill\penalty50\hskip1em\null\nobreak
		\hfill\mbox{\normalfont(#1)}%
		\parfillskip=0pt \finalhyphendemerits=0 \par}%
}

\NewDocumentEnvironment{pquotation}{m}
{\begin{quoting}[
		indentfirst=true,
		leftmargin=\parindent,
		rightmargin=\parindent]\itshape}
	{\bywhom{#1}\end{quoting}}

\definecolor{high}{HTML}{00CC99}  %
\definecolor{low}{HTML}{fff7bc}  %
\newcommand*{\opacity}{90}%
\newcommand*{\minval}{1.0}%
\newcommand*{\maxval}{5.0}%
\newcommand{\gradient}[1]{
    \ifdimcomp{#1pt}{>}{\maxval pt}{#1}{
    \ifdimcomp{#1pt}{<}{\minval pt}{#1}{
         \pgfmathparse{int(round(100*(#1/(\maxval-\minval))-(\minval*(100/(\maxval-\minval)))))}
        \xdef\tempa{\pgfmathresult}
        \cellcolor{high!\tempa!low!\opacity} #1
    }}
 }

\usepackage{url}

\begin{document}

\title{A Novel Loop Fission Technique Inspired by Implicit Computational Complexity%
	\texorpdfstring{%
    	\thanks{
		This research is supported by the \href{https://face-foundation.org/higher-education/thomas-jefferson-fund/}{Thomas Jefferson Fund} of the Embassy of France in the United States and the \href{https://face-foundation.org/}{FACE Foundation}. Th.\ Rubiano and Th.\ Seiller are also supported by the Île-de-France region through the DIM RFSI project \enquote{CoHOp}.}
		}%
	{} %
}
\titlerunning{A Novel Loop Fission Technique Inspired by ICC} %
\author{%
		Clément Aubert\texorpdfstring{\inst{1}\orcidID{0000-0001-6346-3043}}{} \and
		Thomas Rubiano\texorpdfstring{\inst{2}}{} \and
		Neea Rusch\texorpdfstring{\inst{1}\orcidID{0000-0002-7354-5330}}{} \and
    	Thomas Seiller\texorpdfstring{\inst{2,3}\orcidID{0000-0001-6313-0898}}{}
	}
\authorrunning{C. Aubert et al.}

\institute{%
	School of Computer and Cyber Sciences, Augusta University \and
	LIPN – UMR 7030 Université Sorbonne Paris Nord \and
	CNRS
}
\maketitle %

\begin{abstract}
This work explores an unexpected application of Implicit Computational Complexity (ICC) to parallelize loops in imperative programs.
Thanks to a lightweight dependency analysis, our algorithm allows splitting a loop into multiple loops that can be run in parallel, resulting in gains in terms of execution time similar to state-of-the-art automatic parallelization tools when both are applicable.
Our graph-based algorithm is intuitive, language-agnostic, proven correct, and applicable to all types of loops, even if their loop iteration space is unknown statically or at compile time, if they are not in canonical form or if they contain loop-carried dependency.
As contributions we deliver the computational technique, proof of its preservation of semantic correctness, and experimental results to quantify the expected performance gains. Our benchmarks also show that the technique could be seamlessly integrated into compiler passes or other automatic parallelization suites.
We assert that this original and automatable loop transformation method was discovered thanks to the \enquote{orthogonal} approach offered by ICC.
\end{abstract}


\section{%
	Original Approaches to Automatic Parallelization} %
\label{sec:intro}

\subsection{Use Cases for Correct Automatic Parallelization}

The demand for perpetually more performant systems has historically driven innovation in hardware, by increasing numbers of transistors and improving clock speeds, but with Dennard scaling and the end of Moore's law in sight, the focus is steadily shifting toward obtaining gains through high-performance and parallel computing~\cite[Chapter 1]{ppc2022}.
Existing parallel programming APIs, such as OpenMP~\cite{OARB21}, PPL~\cite{ppl2021}, and oneTBB~\cite{onetbb2022}, facilitate this progression; but several outstanding issues remain: classic algorithms are written sequentially without parallelization in mind and require reformatting to fit the parallel paradigm. Suitable sequential programs with opportunity for parallelization must be modified, often manually, by inserting parallelization directives. This is a time-consuming and error-prone process: numerous dependencies and multiple levels of function calls must be analyzed to identify regions that can be safely parallelized, before inserting the correct directives---a challenging task to programmers accustomed to sequential execution environment. Lastly, the state explosion resulting from parallelization makes it impossible to exhaustively test the code running on parallel architectures~\cite%
{baier2008}.

The need for parallelization extends beyond software presently developed: to leverage the potential speedup available on modern hardware, all programs---including legacy software---should instruct the hardware to take advantage of its available processors. This induces demand for automatic transformations of large bodies of software to semantically equivalent parallel programs. Since inserting \emph{incorrect} parallel directives is easy, and can lead to performance degradation or alteration in the program's behavior, \emph{correct} automatic parallelization can be used in lieu of human ingenuity, and benchmarked to approximate the potential gain. A useful automatic parallelization tool needs to deliver two things: a cost-benefit analysis, to show it generates speedup on parallel architectures, and proof of preservation of semantic correctness~\cite[Section 3.5]{chandra2000}.

Compilers offer an ideal integration point for many program analyses and optimizations. Automatic parallelization is already a standard feature in developing industry compilers, optimizing compilers, and specialty source-to-source compilers. Tools that perform local transformations, generally on loops, are frequently conceived---but not necessarily implemented---as compiler passes. How those passes are intertwined with sequential code optimizations can sometimes be problematic~\cite{Bertolacci2018}: as an example, OpenMP directives are by default applied early in the compilation, or even at the front-end, and hence the parallelized source code cannot benefit from sequential optimizations such as unrolling. Furthermore, compilers tend to make conservative choices and often miss opportunities to parallelize, \eg on complex scientific and engineering codes~\cite{Bertolacci2018,Holewinski2012}. 

Given this background, the need to parallelize source code automatically and at scale is evident, but the existing approaches are not perfect. The contribution presented in this paper offers an incremental improvement in this direction: it introduces an automatable and graph-based computational method for semantic-preserving loop transformation. By producing a transformed program, shown to be amiable to integration or pipelining with existing automatic parallelization tools, it fits the described landscape by offering potential to improve versatility and richness of various existing parallel compilation toolchains. 

\subsection{Leveraging ICC %
	 for Correct and Universal Transformation}

The technique presented in this paper is founded on Implicit Computational Complexity (ICC) theory~\cite{DalLago2012a}: 
this work is part of a series~\cite{Aubert2022g} that explores how ICC can provide new---sometimes orthogonal---approaches to problems such as code optimization~\cite{Moyen2017,Moyen2017b} or static analysis~\cite{Aubert2022b}.
Critical to this approach is a strong mathematical backbone that allows to \enquote{embed} in the program itself a guarantee of its resource usage, using \eg bounded recursion~\cite{Bellantoni1992,Leivant1993} or type systems~\cite{Baillot2004,Lafont2004}.
This orthogonal approach sometimes allows avoiding difficulties other techniques must address and offers gain in terms of \eg speed, but possibly at the price of precision~\cite[Sect. C]{Aubert2022b}.

More precisely, the presented technique demonstrates how a dependency analysis mechanism, first introduced in our previous work~\cite{Moyen2017}, can be further leveraged to obtain \emph{loop-level parallelism}: a form of parallelism concerned with extracting parallel tasks from loops.
We identify an original way of performing \emph{loop fission}, an optimization technique that breaks loops into multiple loops with the same condition or index range, each taking only a part of the original loop's body.
Our original technique possesses three notable properties:

\begin{description}
	\item[Suitable to loops with unknown iteration spaces]---program analysis does not require knowing loop iteration space statically nor at compile time, making it applicable to loops not in canonical form, which are often ignored (\autoref{app:sec:comparison}).
	\item[Loop-agnostic]---the technique requires practically no structure from the loops: in particular, they can be \prc|while|, \prc|do ... while| or \prc|for| loops, have arbitrarily complex update and termination conditions, loop-carried dependencies, and arbitrarily deep loop nests.
	\item[Language-agnostic]---the method can be used on any imperative language, making it flexible and suitable for realization and integration with tools and languages ranging from high-level to intermediate representations.
\end{description}

The intent of our work is not to replace polyhedral models~\cite{Karp1967}---that are also pushing to remove some restrictions~\cite{Benabderrahmane2010}---, advanced dependency analysis or tools developed for very precise cases (such as loop tiling~\cite{Bertolacci2018}), nor to build concrete competing implementations.
Our goal is to illustrate how ICC has potential to introduce novel and orthogonal optimization techniques: we take as a positive sign the fact that we can facilitate the discovery of equivalent parallel implementations that are not reachable through a pre-established set of correct transformation rules, as complementary to existing methods. %
This also benefit our approach, making scheduling or optimization of caching out of the scope of this work, since they can be deferred to the tool implementing our algorithm.

\subsection{Our Contribution: From Theory to Benchmarks}
	
	Our contribution spans from theoretical foundations to concrete measurements, to deliver a complete perspective on the design and expected real-time efficiency of the introduced method. We present three contributions:  \begin{enumerate}
 
 \item The design of a loop fission transformation---\autoref{ssec:algo}---that analyzes dependencies of loop condition and body variables; establishes cliques between statements, and splits independent cliques into multiple loops.
 
 \item The correctness proof---\autoref{sec:correctness}---that guarantees the semantic preservation of loop transformation.
 
\item Experimental results---\autoref{sec:quality}---that evaluate the potential gain of the proposed technique, including for loops with unknown iteration spaces, and  demonstrates its integratability with existing parallelization frameworks.

\end{enumerate}

This paper is organized as follows: \autoref{sec:background} defines the %
 \enquote{building blocks} for our program transformation method---its language and how dependencies are computed---, \autoref{sec:algo} details the loop fission algorithm, \autoref{sec:quality} presents the experimental results, and \autoref{sec:conclusion} concludes.

\section{Background: Language and Dependency Analysis}
\label{sec:background}

\subsection{A Simple While Imperative Language With Parallel Capacities}

We work with a simple imperative \texttt{WHILE}-language, with semantics similar to \texttt{C}, extended with a \prc|parallel| command, similar to \eg OpenMP's directives~\cite{OARB21}, allowing to execute its arguments in parallel.
The grammar is given by: 
\begin{align*}
	\textit{var} \Coloneqq & \comm{i} \ | \ \comm{j} \ | \ \hdots \ | \ \comm{s} \ | \ \comm{t} \ | \hdots \ | \ \comm{x_1}\ |\ \comm{x_2}\ | \ \hdots \ | \ \comm{z_n} \ | \ \textit{var}[\textit{exp}] \tag{Variables}\\
	\textit{exp} \Coloneqq & \textit{var} \ |\ \textit{val} \ |\ \textit{op(}\emph{exp}, \hdots,\emph{exp}\texttt{)} \tag{Expression}  \\
	\textit{com} \Coloneqq & \textit{var} \gets \mathit{exp}\ |\ \text{\prc|if| } \textit{exp} \text{ \prc|then| } \textit{com} \text{ \prc|else| } \textit{com}\  |\ \\ 
				           & \text{\prc|while| } \textit{exp} \text{ \prc|do| } \textit{com} \ |\ \text{\prc|use|}(\textit{var}, \hdots, \textit{var}) \ |\ \text{\prc|skip|}\ |\ \\
						   & \textit{com;com}\ |\ 	\text{\prc|parallel|}\{\textit{com}\}\{\textit{com}\}\cdots\{\textit{com}\} \tag{Command}
\end{align*}

A variable represents either an undetermined \enquote{primitive} datatype, \eg not a reference variable, or an array, whose indices are given by an expression. %
An expression is either a variable, a value (\eg integer literal) or the application to expressions of some operator \textit{op}, which can be \eg relational (\eg \texttt{==},  \texttt{<}) or arithmetic  (\eg \texttt{+}, \texttt{-}).
We let \(\comm{V}\) (\resp \(\comm{e}\), \(\comm{C}\)) ranges over variables (\resp expression, command), write \eg \prc|if| \(\comm{e}\) \prc|then| \(\comm{C}\) for \prc|if| \(\comm{e}\) \prc|then| \(\comm{C}\) \prc|else| \(\comm{skip}\), and sometimes replace the semicolon with a new line.
We assume commands to be correct, \eg with operators correctly applied to expressions, no out-of-bounds errors, \etc

A \texttt{WHILE} program is thus a sequence of statements, each
statement being either an \emph{assignment}, a \emph{conditional}, a
\emph{while} loop, a \emph{function call}\footnote{%
	The
	\texttt{use} command represents any command which does not modify its
	variables but use them and should not be moved around carelessly
	(\eg a \prc|printf|). In practice, we currently treat all
	function calls as \texttt{use}, even if the function is
	pure. 
}%
 or a \emph{skip}.
 \emph{Statements} are abstracted into \emph{commands}, which can be a statement, a sequence of commands, or multiple commands to be run in parallel.
The semantics of \prc|parallel| is the following: variables appearing in the arguments are considered local, and the value of a given variable \prc|x| after execution of the \prc|parallel| command is the value of the last modified local variable $\comm{x}$. 
This implies possible race conditions, but our transformation will be robust to those: we will assume given \prc|parallel|-free programs, and will introduce \prc|parallel| commands that either uniformly update the variables accross commands, or update them in only one command.

For convenience we define the following sets of variables.

\begin{definition}
	\label{def:in-out-occ}
	Let $\comm{C}$ be a command, we let  $\Out(\comm{C})$ (\resp $\In(\comm{C})$, \(\Occ(\comm{C})\)) be the set of variables \emph{modified} by (\resp \emph{used} by, \emph{occuring} in) $\comm{C}$, as follows:
	
	\noindent
	\resizebox{1\textwidth}{!}{
		\begin{tabular}{| c || c | c | c |}
			\hline
			$\comm{C}$ & $\Out(\comm{C})$ & $\In(\comm{C})$ & $\Occ(\comm{C})$ \\
			\hline
			\hline 
			\textnormal{\prc|x|} = $\comm{e}$ & \prc|x| & $\Occ(\comm{e})$  & \textnormal{\prc|x|} $\cup \Occ(\comm{e})$\\
			\hline 
			\textnormal{\prc|t[|}$\comm{e_1}$\textnormal{\prc|]|} = $\comm{e_2}$ & \textnormal{\prc|t|} & $ \Occ(\comm{e_1}) \cup \Occ(\comm{e_2})$  & \textnormal{\prc|t|} $\cup \Occ(\comm{e_1}) \cup \Occ(\comm{e_2})$\\
			\hline 
			\textnormal{\prc|if|} $\comm{e}$ \textnormal{\prc|then|} $\comm{C_1}$ \textnormal{\prc|else|}\  $\comm{C_2}$ & $\Out(\comm{C_1}) \cup \Out(\comm{C_2})$ & $\Occ(\comm{e}) \cup \In(\comm{C_1}) \cup \In(\comm{C_2})$ & $\Occ(\comm{e}) \cup \Occ(\comm{C_1}) \cup \Occ(\comm{C_2})$\\ 
			\hline
			\textnormal{\prc|while|} $\comm{e}$ \textnormal{\prc|do|} $\comm{C}$ & $\Out(\comm{C})$ & $\Occ(\comm{e}) \cup \In(\comm{C})$ & $\Occ(\comm{e}) \cup \Occ(\comm{C})$  \\
			\hline
			\textnormal{\prc|use|}(\textnormal{\prc|x|}$_1$, $\hdots$, \textnormal{\prc|x|}$_n$) & $\emptyset$ & $\{\textnormal{\prc|x|}_1, \hdots, \textnormal{\prc|x|}_n\}$ & $\{\textnormal{\prc|x|}_1, \hdots, \textnormal{\prc|x|}_n\}$\\
			\hline
			\textnormal{\prc|skip|} & $\emptyset$ & $\emptyset$ & $\emptyset$ \\
			\hline
			$\comm{C_1};\comm{C_2}$ & $\Out(\comm{C_1}) \cup \Out(\comm{C_2})$ & $\In(\comm{C_1}) \cup \In(\comm{C_2})$ & $\Occ(\comm{C_1}) \cup \Occ(\comm{C_2})$\\
			\hline
		\end{tabular}
	}
for         $ \begin{aligned}[t]
\Occ(\textnormal{\prc|x|}) &= \textnormal{\prc|x|} 
&&&  
\Occ(\textnormal{\prc|t[|$\comm{e}$\prc|]|}) &= \textnormal{\prc|t|} \cup \Occ(\comm{e})\\
\Occ(\textit{val}) & = \emptyset
&&& 
\Occ(\textnormal{\prc|op(|$\comm{e_1}, \hdots,\comm{e_n}$\prc|)|}) &= \Occ(\comm{e_1}) \cup \cdots \cup \Occ(\comm{e_n})
\end{aligned}$
\end{definition}

Our treatment of arrays is an over-approximation: we consider the array as a single entity, and that changing one value in it changes it completely.
This is however satisfactory: since we will not split loop \enquote{horizontally} (\eg splitting the iteration space between threads) but \enquote{vertically} (\eg splitting the tasks between threads), we want each thread in the \prc|parallel| command to have \enquote{full control} of the array it modifies, and not to synchronize its writes with other commands.

\subsection{Datalow Graphs for Loop Dependency Analysis}

The loop transformation algorithm relies fundamentally on its ability to analyze data-flow dependencies between loop condition and variables in the loop body, to identify opportunities for loop fission. In this section we sketch the principles of this dependency analysis, founded on the theory of \emph{data-flow graphs}, and how it maps to the presented \texttt{WHILE}-language. This dependency analysis was influenced by large body of works related to static analysis \cite{Abel20002,Kristiansen2005b,Lee2001}, semantics \cite{Laird2013,%
Seiller2016} and optimization~\cite{Moyen2017}; but is presented here in self-contained and compact manner.%

\subsubsection{Definition of Data-Flow Graphs}
\label{DFG}

A data-flow graph for a given command $\comm{C}$ is a weighted relation on the set \(\Occ(\comm{C})\).
Formally, this is represented as a matrix over a semi-ring, with the implicit choice of a denumeration of \(\Occ(\comm{C})\).

\begin{definition}
	A \emph{data-flow graph} (\DFG) for a command $\comm{C}$ is a $|\Occ(\comm{C})|\times |\Occ(\comm{C})|$ matrix over a fixed semi-ring $(\mathcal{S},+,\times)$, with $|\Occ(\comm{C})|$ the cardinal of \(\Occ(\comm{C})\).
	We write $\dfg{C}$ the \DFG of $\comm{C}$, and explain how to construct it below. %
\end{definition}

To avoid resizing matrices whenever additional variables are considered, we identify $\dfg{C}$ with its embedding in a larger matrix, \ie we will abusively call the
\DFG of $\comm{C}$ any matrix of the form
$
	\dfg{C} \oplus \Id
$, 
implicitly viewing the additional rows/columns as variables not in $\Occ(\comm{C})$.
We will use weighted relations, or weighted bi-partite graphs, to illustrate these matrices. 
Examples will use the semi-ring $(\{0,1,\infty\},\max,\times)$, which is the specific semiring considered in later sections to represent dependencies:
$\infty$ represents \emph{dependence},
$1$ represents \emph{propagation}, and  $0$ represents \emph{reinitialization} or independence.
\autoref{fig:dependences} introduces these notions and the graphical conventions used throughout this paper.
Note that in the case of dependencies, $\In(\comm{C})$ is exactly the set of variables that are source of a \enquote{dependence} arrow, while $\Out(\comm{C})$ is the set of variables that either are targets of dependence arrows or were reinitialized.

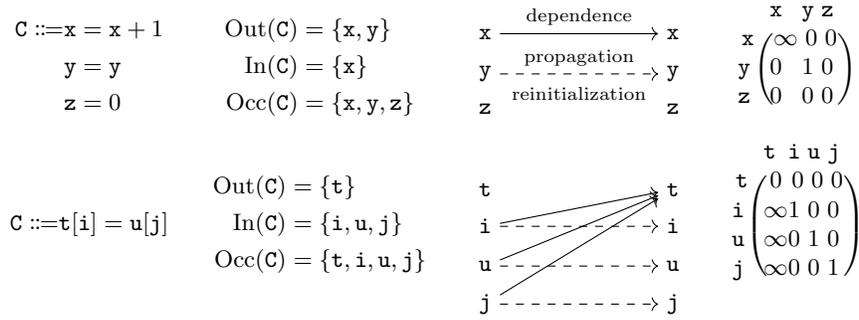
\begin{figure}
	{
	\addtolength\tabcolsep{7pt}
	\centering
	
	\begin{tabular}{c c c c}
	$\begin{aligned}
		\comm{C} \Coloneqq & \comm{x=x+1}\\
						   &\comm{y=y} \\
						   & \comm{z=0}
	\end{aligned}$
	
	& 
	
	$\begin{aligned}
	\Out(\comm{C}) & = \{\comm{x, y}\}\\
	\In(\comm{C})  & = \{\comm{x}\}\\
	\Occ(\comm{C}) & = \{\comm{x, y, z}\}
	\end{aligned}$
	
	&
	
	\begin{tikzpicture}[anchor=base, baseline=-1.4em]
	\node (x0) at (1.5,0) {$\comm{x}$};
	\node (x1) at (1.5,-.5) {$\comm{y}$};
	\node (x2) at (1.5,-1) {$\comm{z}$};
	\node (y0) at (4,0) {$\comm{x}$};
	\node (y1) at (4,-.5) {$\comm{y}$};
	\node (y2) at (4,-1) {$\comm{z}$};
	\draw [->] (x0) -- node[above, font=\scriptsize, midway]{dependence} (y0); 
	\draw [dashed, ->] (x1) -- node[above, font=\scriptsize, midway]{propagation} (y1); 
	\draw [white] (x2) -- node[above, font=\scriptsize, midway, black]{reinitialization} (y2); 
	\end{tikzpicture}
	
	&
		\begin{blockarray}{l c c c}
		\hspace{1.3em} & $\comm{x}$ & $\comm{y}$ & $\comm{z}$\\
		\begin{block}{l ( c c c )}
		$\comm{x}$ & $\infty$ & 0 & 0 \\
		$\comm{y}$ & 0 & 1 & 0 \\
		$\comm{z}$ & 0 & 0 & 0\\
		\end{block}
		\end{blockarray}
	\\
		$\begin{aligned}
		\comm{C} \Coloneqq & \comm{t[i] =u[j]}\\
		\end{aligned}$
	
		& 
		
		$\begin{aligned}
		\Out(\comm{C}) & = \{\comm{t}\}\\
		\In(\comm{C})  & = \{\comm{i, u, j}\}\\
		\Occ(\comm{C}) & = \{\comm{t, i, u, j}\}
		\end{aligned}$
		
		&
		
		\begin{tikzpicture}[anchor=base, baseline=-1.4em]
		\node (x0) at (1.5,0) {$\comm{t}$};
		\node (x1) at (1.5,-.5) {$\comm{i}$};
		\node (x2) at (1.5,-1) {$\comm{u}$};
		\node (x3) at (1.5,-1.5) {$\comm{j}$};
		\node (y0) at (4,0) {$\comm{t}$};
		\node (y1) at (4,-.5) {$\comm{i}$};
		\node (y2) at (4,-1) {$\comm{u}$};
		\node (y3) at (4,-1.5) {$\comm{j}$};
		\draw [->] (x1) --  (y0); 
		\draw [->] (x2) --  (y0); 
		\draw [->] (x3) --  (y0); 
		\draw [dashed, ->] (x1) -- (y1); 
		\draw [dashed, ->] (x2) -- (y2); 
		\draw [dashed, ->] (x3) -- (y3); 
		\end{tikzpicture}
		
		&
		\begin{blockarray}{l c c c c}
			\hspace{1.3em} & $\comm{t}$ & $\comm{i}$ & $\comm{u}$ & $\comm{j}$\\
			\begin{block}{l ( c c c c )}
				$\comm{t}$ & 0 & 0 & 0 & 0\\
				$\comm{i}$ & $\infty$ & 1 & 0 & 0\\
				$\comm{u}$ & $\infty$ & 0 & 1 & 0\\
				$\comm{j}$ & $\infty$ & 0 & 0 & 1\\
			\end{block}
		\end{blockarray}

	\end{tabular}
	}

	\caption{Program examples, sets, and representations of their dependences\label{Fig_threecases}}
	\label{fig:dependences}
\end{figure}

\subsection{Constructing Data-Flow Graphs (\DFGs)}
\label{ssec:constructionDFGs}

The \DFG of a command is computed by induction on the structure of the command.

\subsubsection{Base cases (assignment, skip, use)}%
The \DFG for assignments are obtained by straightforward generalization of the cases %
 illustrated in \autoref{Fig_threecases}, and	$\dfg{\mathtt{skip}}$ is the \enquote{empty matrix} with $0$ rows and columns\footnote{Identifying the \DFG with its embeddings, it is hence the identity matrix of any size.}.

To account for \prc|use|(\prc|x|$_1$, $\hdots$, \prc|x|$_n$), we 
introduce a variable $\mathtt{e}$---standing for \emph{effect}---not being part of the language, and let $\dfg{\text{\prc|use|(\prc|x|$_1$, $\hdots$, \prc|x|$_n$})}$ be the matrix with coefficients
	from each \prc|x|$_{i}$ and $\mathtt{e}$ to $\mathtt{e}$ equal to $\infty$, and $0$
	coefficients otherwise. %

\subsubsection{Composition and multipaths}

The definition of \DFG for a (sequential)
\emph{composition} of commands is an abstraction that allows treating a block of statements as one command with its own \DFG.

\begin{definition}
	$\dfg{\comm{C_{1}};\comm{C_{2}};\dots;\comm{C_{n}}}$
	is %
	the matrix product
	$\dfg{C_{1}}\dfg{C_{2}}\cdots\dfg{C_{n}}$.
\end{definition}

For two graphs, %
the product of their matrices of weights is represented in a standard way, as a graph of length 2 paths; %
as illustrated in \autoref{fig:composition}.

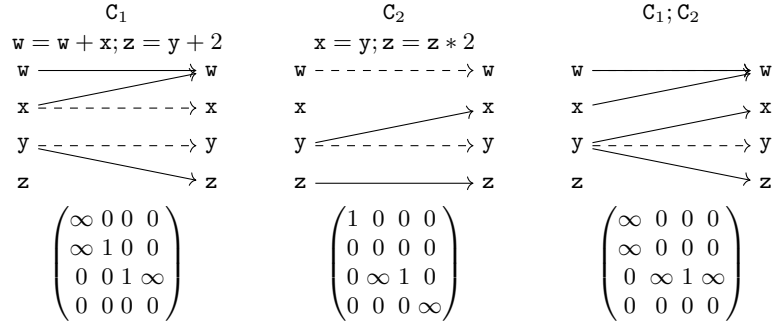
\begin{figure}%
\begin{center}
	{
	\addtolength\tabcolsep{10pt}
	\centering

	\begin{tabular}{c c c }
			\(\comm{C_{1}}\) &\(\comm{C_{2}}\) & \(\comm{C_{1}};\comm{C_2}\)\\
			$\comm{w}=\comm{w}+\comm{x};\comm{z}=\comm{y}+2$ & \(\comm{x}=\comm{y};\comm{z}=\comm{z}*2\) \\
			\begin{tikzpicture}
			\node (x0) at (1.5,0) {$\comm{w}$};
			\node (x1) at (1.5,-.5) {$\comm{x}$};
			\node (x2) at (1.5,-1) {$\comm{y}$};
			\node (x3) at (1.5,-1.5) {$\comm{z}$};
			\node (y0) at (4,0) {$\comm{w}$};
			\node (y1) at (4,-.5) {$\comm{x}$};
			\node (y2) at (4,-1) {$\comm{y}$};
			\node (y3) at (4,-1.5) {$\comm{z}$};
			\draw [->] (x0) --  (y0); 
			\draw [->] (x1) --  (y0); 
			\draw [->] (x2) --  (y3); 
			\draw [dashed, ->] (x1) --  (y1); 
			\draw [dashed, ->] (x2) --  (y2); 
			\end{tikzpicture}
			&
			\begin{tikzpicture}
			\node (x0) at (1.5,0) {$\comm{w}$};
			\node (x1) at (1.5,-.5) {$\comm{x}$};
			\node (x2) at (1.5,-1) {$\comm{y}$};
			\node (x3) at (1.5,-1.5) {$\comm{z}$};
			\node (y0) at (4,0) {$\comm{w}$};
			\node (y1) at (4,-.5) {$\comm{x}$};
			\node (y2) at (4,-1) {$\comm{y}$};
			\node (y3) at (4,-1.5) {$\comm{z}$};
			\draw [->] (x2) --  (y1); 
			\draw [->] (x3) --  (y3); 
			\draw [dashed, ->] (x0) --  (y0); 
			\draw [dashed, ->] (x2) --  (y2); 
			\end{tikzpicture}
			& 
			\begin{tikzpicture}
			\node (x0) at (1.5,0) {$\comm{w}$};
			\node (x1) at (1.5,-.5) {$\comm{x}$};
			\node (x2) at (1.5,-1) {$\comm{y}$};
			\node (x3) at (1.5,-1.5) {$\comm{z}$};
			\node (y0) at (4,0) {$\comm{w}$};
			\node (y1) at (4,-.5) {$\comm{x}$};
			\node (y2) at (4,-1) {$\comm{y}$};
			\node (y3) at (4,-1.5) {$\comm{z}$};
			\draw [->] (x0) --  (y0); 
			\draw [->] (x1) --  (y0); 
			\draw [->] (x2) --  (y1); 
			\draw [->] (x2) --  (y3); 
			\draw [dashed, ->] (x0) --  (y0); 
			\draw [dashed, ->] (x2) --  (y2); 
			\end{tikzpicture}
			\\
			\(
			\begin{pmatrix}
			\infty & 0 & 0 & 0 \\
			\infty & 1 & 0 & 0 \\
			0 & 0 & 1 & \infty \\
			0 & 0 & 0 & 0\\
			\end{pmatrix}
			\) & 
			$\begin{pmatrix}
			1 & 0 & 0 & 0 \\
			0 & 0 & 0 & 0 \\
			0 & \infty & 1 & 0 \\
			0 & 0 & 0 & \infty\\
			\end{pmatrix}
			$
			& 
			$\begin{pmatrix}
			\infty & 0 & 0 & 0 \\
			\infty & 0 & 0 & 0 \\
			0 & \infty & 1 & \infty \\
			0 & 0 & 0 & 0\\
			\end{pmatrix}
			$
		\end{tabular}
	}
\end{center}

	\caption{
		Data-Flow Graph of Composition. 
	}\label{fig:composition}
\end{figure}%

\subsubsection{Conditionals.}\label{conditional}

We define the \DFG of \prc|if| \(\comm{e}\) \prc|then| \(\comm{C_{1}}\) \prc|else| \(\comm{C_{2}}\) from
the \DFG of the commands $\comm{C_{1}}$ and $\comm{C_{2}}$. First, consider a situation 
where both commands $\comm{C_{1}}$ and $\comm{C_{2}}$ are potentially executed. 
In that case, the statement should be represented by the overapproximation 
$\dfg{C_{1}}+\dfg{C_{2}}$.
However, all the modified variables in \(\comm{C_1}\) and \(\comm{C_2}\) (\eg \(\Out(\comm{C_1}) \cup \Out(\comm{C_2})\)) depends on the variables used in \(\comm{e}\) (\eg occuring in \(\comm{e}\)).
For this reason, letting \(E\) be the vector with coefficient equal to \(\infty\) for the variables in \(\comm{e}\) and \(0\) for all the other variables,
\(O\) be the vector containing the variables in \(\Out(\comm{C_1}) \cup \Out(\comm{C_2})\), and \((\cdot)^t\) be the matrix transpose, we define $\corr{\comm{e}} = (EO)^t$, and have---as illustrated previously~\cite[Fig. 3]{Moyen2017}---:

\begin{definition}
 $\dfg{\textnormal{\prc|if| \(\comm{e}\) \prc|then| \(\comm{C_{1}}\) \prc|else| \(\comm{C_{2}}\)}}=\dfg{C_{1}}+\dfg{C_{2}}+ \corr{\comm{e}}$.
\end{definition}

\subsubsection{While Loops.}
To define the \DFG of a command \prc|while| $\comm{e}$ \prc|do| $\comm{C}$
from $\dfg{C}$, we need, as for conditionals, the \emph{loop correction} $\corr{\comm{e}}$, to account for the fact that all the modified variables in \(\comm{C}\) depends on the variables used in \(\comm{e}\):

\begin{definition}
\(\dfg{\textnormal{\prc|while| $\comm{e}$ \prc|do| $\comm{C}$}} = \dfg{C} + \corr{\comm{e}}\).
\end{definition}

This is different from our previous treatment of \prc|while| loop~\cite[Definition 5]{Moyen2017}, that required to compute the transitive closure of $\dfg{C}$: for this particular transformation, this is not needed, as all the relevant dependencies are obtained immediately---this also guarantee that loop-carried dependencies~\cite{enwiki:1080087398} do not refrain from parallelizing the body of the loops, in our analysis.

\section{Loop Fission Algorithm}
	\label{sec:algo}
\subsection{Algorithm, Presentation and Intuition}
\label{ssec:algo}

Leveraging the presented dependency analysis, we can now define the specifics of our loop transformation technique, given in \autoref{algo} and explained below.

      \begin{algorithm}%
       \caption{Loop fission}
       \label{algo}

\hspace*{\algorithmicindent/2} \textbf{Input:} loop $w = \{\texttt{while} \in \texttt{WHILE} \}$
\begin{algorithmic}
\State $vertices \leftarrow$ empty list
\State $In, Out \leftarrow$ of loop body variables
\State $parent \leftarrow $ parent statement of $w$ 
\ForAll {$in, out \in$ c($In$, $Out$)}
\If {$in $ is dependency of $out$}
   \State $vertices \mathrel{+}=$  vertex($in$, $out$) 
\EndIf
\EndFor
\ForAll {$cond \in$ loop conditions}
   \State $vertices \mathrel{+}=$  ($cond$, body stmts) 
\EndFor
\State $digraph$ $\leftarrow$ from($vertices$)
\State $sccs$ $\leftarrow$ reduce($digraph$)
\State $dag$ $\leftarrow$ condensation($digraph$, $sccs$)
\State $subgraphs$ $\leftarrow$ from $dag$
\If {$\vert subgraphs \vert > 1$}
\State $parent$ remove($w$)
\ForAll {$scc \in sccs$}
   \State $parent$ insert loop from scc 
\EndFor 
\EndIf 
\end{algorithmic}
      \end{algorithm}

 Given a loop 
$\comm{C}:= \textnormal{\prc|while| $\comm{e}$ \prc|do| $\{\comm{C}_1; \cdots; \comm{C}_n\}$}$, 
we first compute $\dfg{C_{1};\dots;C_{n}}$, and add the loop correction $\corr{\comm{e}}$. %
The dependence graph of the loop $\comm{C}$ is then defined as as the graph where the set of vertices is the set of commands $\{\comm{C}_1; \cdots; \comm{C}_n\}$, and there exists a directed edge from $\comm{C}_i$ to $\comm{C}_j$ if and only if there exists variables \(\text{\prc|x|}\in \Out(\comm{C}_j)$ and \(\text{\prc|y|}\in \In(\comm{C}_i)\) such that $\dfg{C}(\text{\prc|x|},\text{\prc|y|})=\infty$.
Note that all the commands are the sources of dependence edges whose target is the commands modifying the variables occuring in $\comm{e}$ thanks to the correction. %
      
The remainder of the loop transforming principle is simple: for each loop in the analyzed program, after evaluating data dependencies in the loop condition and variables in the loop body, it produces a graph representing the dependencies between commands; then determines the cliques in the graph and forms \emph{strongly connected components} (\sccs); the \sccs are separated into subgraphs to produce the final split loops, when applicable, and contain a copy of the loop header and update commands.

\begin{figure}
\begin{minipage}[c]{0.29\linewidth}		
\begin{algorithmic}
	\State $j \gets \textbf{pow}(10, 6)$
	\State $x_1 \gets 1$
	\State $i \gets 1$
	\While{$i \neq j$}
	\State $x_1 \gets x_1 + y_1 + x_2 + i$
	\State $y_1 \gets y_1+i$
	\State $y_2 \gets y_2\times y_1$
	\State $s[i] \gets x_1$
	\State $x_2 \gets x_1+s[i]$
	\State $u[i] \gets y_2$
	\State $t[i] \gets y_2 \times y_2$  
	\State $i \gets i+1$
	\EndWhile
	\Use{$x_1$}
\end{algorithmic}
\end{minipage}
\begin{minipage}[c]{0.7\linewidth}
	\begin{tikzpicture}[x=1.35cm,y=1.5cm]
	\node[rectangle,dotted,rounded corners=3] (1) at (-1,0) {\footnotesize{$x_1 \gets x_1 + y_1 + x_2 + i$}};
	\node[draw,rectangle,dotted,rounded corners=3] (2) at (2,0) {\footnotesize{$y_1 \gets y_1+i$}};
	\node[draw,rectangle,dotted,rounded corners=3] (3) at (2,-.8) {\footnotesize{$y_2 \gets y_2\times y_1$}};
	\node[rectangle,dotted,rounded corners=3] (4) at (-1.8,-1) {\footnotesize{$s[i] \gets x_1$}};
	\node[rectangle,dotted,rounded corners=3] (5) at (0,-1) {\footnotesize{$x_2 \gets x_1+s[i]$}};
	\node[draw,rectangle,dotted,rounded corners=3] (6) at (3,-1.5) {\footnotesize{$u[i] \gets y_2$}};
	\node[draw,rectangle,dotted,rounded corners=3] (7) at (1,-1.5) {\footnotesize{$t[i] \gets y_2 \times y_2$}};
	\node[draw,rectangle,dotted,rounded corners=3] (8) at (0,1) {\footnotesize{$i \gets i+1$}};
	
	\draw[->] (7)--(3);
	\draw[->] (1)--(2);
	\draw[->] (3)--(2);
	\draw[->] (6)--(3);
	\draw[->] (4)--(1);
	\draw[->] (5) to [bend right](1);
	\draw[->] (5)--(4);
	\draw[->] (1) to[bend right] (5) ;
	\draw[->] (1)--(8);
	\draw[->] (2)--(8);
	\draw[->] (8) to [out=60,in=120,looseness=8] (8);
	
	\draw[dotted,rounded corners=5] (-2.35,-1.15) rectangle (0.8,0.15);
	\end{tikzpicture}
\end{minipage}
\caption{\texttt{WHILE}-language program and its corresponding dependency graph.}\label{ex:prog1} 
\end{figure}
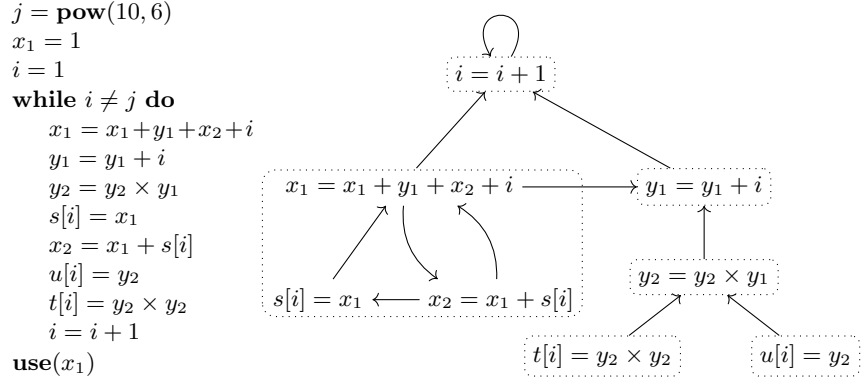

Let us consider a program and its corresponding graph, in \autoref{ex:prog1}, where strongly connected components (\sccs) are shown as dotted rectangles: the \emph{condensation graph} is the graph whose vertices are \sccs and edges are the edges whose source and target belong to distinct \sccs.
Our algorithm splits this graph into its branches, introducing duplications, and outputs a parallelizable version of the initial program as follows. It first splits the graph into branches: for this, we work not with the dependency graph directly but with the acyclic graph obtained from its decomposition into \sccs.
This graph is known as the \emph{condensation} of the initial graph.
The splitting algorithm produces automatically the corresponding program and its \emph{covering} (\autoref{def:covering}), presented in Fig. \ref{ex:optimized} and \ref{ex:optimized-bis}.

\begin{figure}
	\begin{algorithmic}
	\State $j \gets \textbf{pow}(10, 6)$
	\State $x_1 \gets 1$
	\State $i \gets 1$
	\State \prc|parallel|   \Comment{A private copy of \(i\), \(y_1\) and \(y_2\) needs to be given to each loop below.}
	\vspace{-1em}
	\begin{multicols}{3}
		\noindent $\overbrace{\hspace{10em}}$
		\While{$i \neq j$} %
		\State $x_1 \gets x_1 + y_1 + x_2 + i$
		\State $y_1 \gets y_1+i$
		\State $s[i] \gets x_1$
		\State $x_2 \gets x_1+s[i]$
		\State $i \gets i+1$
		\EndWhile
		\vspace{-.7em}
		
		\noindent $\underbrace{\hspace{10em}}$
		
		\noindent $\overbrace{\hspace{10em}}$
		\While{$i \neq j$} %
		\State $y_1 \gets y_1+i$
		\State $y_2 \gets y_2\times y_1$
		\State $u[i] \gets y_2$
		\State $i \gets i+1$
		\EndWhile
		\vspace{-.7em}
		
		\noindent $\underbrace{\hspace{10em}}$
		
		\noindent $\overbrace{\hspace{10em}}$
		\While{$i \neq j$}%
		\State $y_1 \gets y_1+i$
		\State $y_2 \gets y_2\times y_1$
		\State $t[i] \gets y_2 \times y_2$  
		\State $i \gets i+1$
		\EndWhile 
		\vspace{-.7em}
		
		\noindent $\underbrace{\hspace{10em}}$
	\end{multicols}
	\Use{$x_1$}
	\Comment{The copies of \(i\), \(y_1\) and \(y_2\) can be destroyed, all have the same values.}
\end{algorithmic}
 \caption{Transformed program}
\label{ex:optimized}
\end{figure}

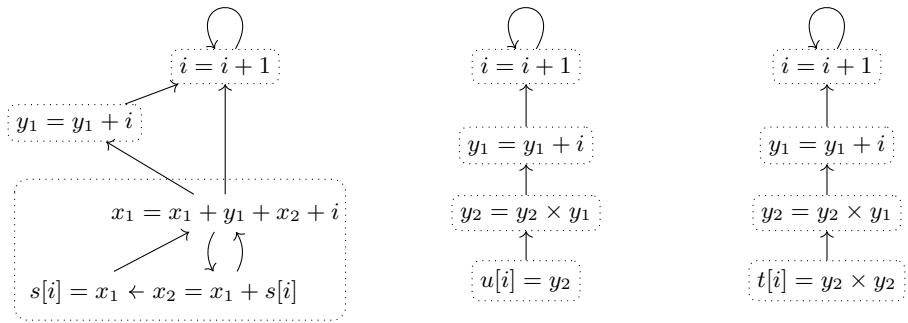
\begin{figure}
	\begin{tikzpicture}[x=2cm,y=1.5cm]
	\node[rectangle,dotted,rounded corners=3] (1) at (0,-.8) {\footnotesize{$x_1 \gets x_1 + y_1 + x_2 + i$}};
	\node[rectangle,dotted,rounded corners=3] (4) at (-1,-1.5) {\footnotesize{$s[i] \gets x_1$}};
	\node[rectangle,dotted,rounded corners=3] (5) at (0,-1.5) {\footnotesize{$x_2 \gets x_1+s[i]$}};
	\node[draw,rectangle,dotted,rounded corners=3] (8) at (0,.5) {\footnotesize{$i \gets i+1$}};
	\node[draw,rectangle,dotted,rounded corners=3] (2) at (-1,0) {\footnotesize{$y_1 \gets y_1+i$}};

	\node[draw,rectangle,dotted,rounded corners=3] (2bis) at (4,-0.2) {\footnotesize{$y_1 \gets y_1+i$}};
	\node[draw,rectangle,dotted,rounded corners=3] (8bis) at (4,.5) {\footnotesize{$i \gets i+1$}};
	\node[draw,rectangle,dotted,rounded corners=3] (3bis) at (4,-.8) {\footnotesize{$y_2 \gets y_2\times y_1$}};
	\node[draw,rectangle,dotted,rounded corners=3] (7) at (4,-1.4) {\footnotesize{$t[i] \gets y_2 \times y_2$}};
	
	\node[draw,rectangle,dotted,rounded corners=3] (6) at (2,-1.4) {\footnotesize{$u[i] \gets y_2$}};
	\node[draw,rectangle,dotted,rounded corners=3] (2ter) at (2,-0.2) {\footnotesize{$y_1 \gets y_1+i$}};
	\node[draw,rectangle,dotted,rounded corners=3] (3ter) at (2,-.8) {\footnotesize{$y_2 \gets y_2\times y_1$}};
	\node[draw,rectangle,dotted,rounded corners=3] (8ter) at (2,.5) {\footnotesize{$i \gets i+1$}};
			
	\draw[->] (7)--(3bis);
	\draw[->] (1)--(2);
	\draw[->] (3bis)--(2bis);
	\draw[->] (3ter)--(2ter);
	\draw[->] (6)--(3ter);
	\draw[->] (4)--(1);
	\draw[->] (5) to [bend right] (1);
	\draw[->] (5)--(4);
	\draw[->] (1) to [bend right] (5);
	\draw[->] (1)--(8);
	\draw[->] (2)--(8);
	\draw[->] (2bis)--(8bis);
	\draw[->] (2ter)--(8ter);
	\draw[->] (8) to [out=60,in=120,looseness=8] (8);
	\draw[->] (8bis) to [out=60,in=120,looseness=8] (8bis);
	\draw[->] (8ter) to [out=60,in=120,looseness=8] (8ter);	
	\draw[dotted,rounded corners=5] (-1.4,-.5) rectangle (0.8,-1.75);
	\end{tikzpicture}

 \caption{Transformed program---covering}
\label{ex:optimized-bis}
\end{figure}

\subsection{Correctness of the Algorithm}
\label{sec:correctness}

Formally, the transformation is described by means of particular kinds of \emph{coverings} of the dependency graph. Let us start with a formal definition.

\begin{definition}[\protect{\cite{Chung1980}}]
	\label{def:covering}
A covering of a (directed) graph $G$ is a collection of subgraphs $G_1,G_2,\dots,G_k$ such that $G=\cup_{i=1}^k G_i$.

A \emph{saturated covering} is a covering $G_1,G_2,\dots,G_k$ such that for all edge in $G$ with source in $G_i$, its target belongs to $G_i$ as well.
\end{definition}

Given a loop $\comm{C}:= \textnormal{\prc|while| $\comm{e}$ \prc|do| $\{\comm{C}_1; \cdots; \comm{C}_n\}$}$,  we explained how one can compute its dependency graph $D(\mathtt{C})$ on the set of vertices $\{1,2,\dots,n\}$ representing the commands in the body of $\mathtt{C}$. Given a saturated covering $G_1,G_2,\dots,G_k$ of the graph $G$, we can define a sequence of loops and prove that the semantics of the initial loop is preserved.
This is done by showing that for any variable $\comm{x}$ appearing in the initial loop, its final value is unchanged%
.

\begin{definition}
	Let $\comm{C}:= \textnormal{\prc|while| $\comm{e}$ \prc|do| $\{\comm{C}_1; \cdots; \comm{C}_n\}$}$  be a command, $D(\mathtt{C})$ its dependency graph, and $G_1,G_2,\dots,G_k$ a saturated covering of the graph $G$. We define $\mathtt{\tilde{C}}=\mathtt{C}^1;\dots;\mathtt{C}^k$ where the command $\mathtt{C}^j$ is defined from $G_j$ by 
	$\comm{C}^j = \textnormal{\prc|while| $\comm{e}$ \prc|do| $\{\comm{C}_{i_1}; \cdots; \comm{C}_{i_m}\}$}$ where $\{i_1,\dots,i_m\}$ is the set of vertices of $G_j$.
\end{definition}

\begin{theorem}
	\label{thm}
The transformation $\mathtt{C}\rightsquigarrow\mathtt{\tilde{C}}$ preserves the semantic.
\end{theorem}

\begin{proof}[Proof (sketch)]
We show that for every variable $\comm{x}$, the value of $\comm{x}$ after the execution of $\mathtt{C}$ is equal to the value of $\comm{x}$ after the execution of $\mathtt{\tilde{C}}$. Variables are considered local to each loop $\mathtt{C}^j$ in $\mathtt{\tilde{C}}$, so we need to avoid race condition. To do so, we prove the following more precise result: for each variable $\comm{x}$ and each loop $\mathtt{C}^j$ in $\mathtt{\tilde{C}}$ in which the value of $\comm{x}$ is modified, the value of $\comm{x}$ after executing $\mathtt{C}$ is equal to the value of $\comm{x}$ after executing $\mathtt{C}^j$.

The previous claim is then straightforward to prove, based on the property of the covering. One shows by induction on the number of iterations $k$ that for all the variables $x_1,\dots,x_h$ appearing in $\mathtt{C}^j$, the values of $x_1,\dots,x_h$ after $k$ loop iterations of $\mathtt{C}^j$ are equal to the values of $x_1,\dots,x_h$ after $k$ loop iterations of $\mathtt{C}$. Note some other variables may be affected by the latter but the variables $x_1,\dots,x_h$ do not depend on them (otherwise, they would also appear in $\mathtt{C}^j$ by definition of the dependence graph and the covering).
\end{proof}

\section{Experimental Results} 
\label{sec:quality}

This section aims at experimentally substantiating two claims:
\begin{enumerate}
\item Our algorithm can parallelize loops that are completely ignored by---to our knowledge (\autoref{app:sec:comparison})---all the other automatic parallelization tools, and result in appreciable gain (\autoref{sec:casestudy}),
\item Our code transformation provides a gain similar to the automatic parallelization tool \texttt{AutoPar-Clava}---which \enquote{compare\textins{s} favorably with closely related auto-parallelization compilers}~\cite[p.~1]{Arabnejad2020}---when both are applicable, and can be integrated in automatic parallelization pipelines (\autoref{ssec:par-techniques}).
\end{enumerate}

Taken together, those results confirm that our original loop fission can easily be integrated in pre-existing tools and improve the performances of the resulting code.
We begin by briefly explaining our benchmarking strategy.

\subsection{Benchmarking Strategy}

The primary goal of our benchmarking strategy was to evaluate the potential performance gain using the described algorithm and parallelization.
The PolyBench/C suite~\cite{polybenchc} was selected for this purpose because it contains programs in the \texttt{C} programming language, which naturally maps to the syntax of \texttt{WHILE} presented in \autoref{sec:background}, the \texttt{parallel} command being represented as OpenMP directives. The suite has been  crafted to offer different opportunities for loop transformations, and comes with built-it timing utilities, which we use, to obtain accurate and comparable results. One drawback is the suite's sole focus on \enquote{canonical} \texttt{for} loops, which prevented evaluating transformations of other kinds of loops that our algorithm can split.
This issue was resolved by introducing an additional case study, measuring  the performance of parallelized \prc|while| loops%
, discussed in \autoref{sec:casestudy}.

To avoid evaluation of non-transformable programs, the suite was then reduced to the programs that were either suitable for loop fission or already had the intended form. There were 6 such programs. Next these programs were transformed manually to their post-fission form. Since the proposed technique also involves parallelization, we defined two parallelization strategies: one where OpenMP directives were inserted manually, and another using an automated parallelizing source-to-source compiler, \texttt{AutoPar-Clava}~\cite{Arabnejad2020}. The two approaches are complementary and serve different purposes: the manual method enables finding optimal directives; the automatic approach shows these tools can be pipelined to obtain fully automatic parallelizing compilation toolchain.

The PolyBench/C benchmark timing script runs each program 5 times, taking the average of 3 runs. The script also measure variance between the averaged times, which in this study was constrained never to exceed 5\%.
Speedup is the ratio of sequential and parallel executions, $S = T_{\text{Seq}}/T_{\text{Par}}$, where a value greater than 1 indicates parallel is outperforming the sequential execution. In presentation of these results, the original sequential programs are always considered the baseline.

The benchmarks were ran on multiple \texttt{gcc} compiler optimization levels (\texttt{O0-O3}), and data sizes supplied with the suite (\texttt{MINI} - \texttt{EXTRALARGE}) using a Linux 4.19.0-20-amd64 \#1 SMP Debian 4.19.235-1 (2022-03-17) x86\_64 GNU/Linux  machine, with 4 Intel(R) Core(TM) i5-6300U CPU @ 2.40GHz processors.
Our open source benchmarking is available at \href{https://github.com/statycc/icc-fission}{https://github.com/statycc/icc-fission}.

\subsection{Case Study: Parallel \prc|while| Loops}
\label{sec:casestudy}

The agnostic treatment of the various kinds of loops, including loops with unknown iteration spaces, are some of the highlights of the presented technique. It is however difficult to compare this approach to other techniques, because most loop transformation and parallelization tools focus only on \prc|for| loops and PolyBench/C does not include \prc|while| loop programs. The difficulty parallelizing \prc|while| loops arises from the need to synchronize evaluation of the loop recurrence and termination condition, with improper synchronization resulting in overshooting the iterations~\cite{Rauchwerger1995}; rendering such loops effectively sequential.
Our technique addresses this challenge by recognizing the independence between loops resulting from loop fission, thus producing parallelizable loop chains. A good candidate for demonstrating this concept is the benchmark program \href{https://github.com/statycc/icc-fission/blob/2c4f131ccfcce1332e08ed7af0f7cfd98737d546/original/3mm.c}{\texttt{3mm}}, from which we constructed a semantically equivalent program using \prc|while| loops, 
\href{https://github.com/statycc/icc-fission/blob/2c4f131ccfcce1332e08ed7af0f7cfd98737d546/case_study-b/3mm_while.c}{\texttt{3mm\_while}}. %

In general, special care is needed when inserting parallelization directives for loops with unknown iteration spaces.
Use of \texttt{single} directive prevents overshooting the loop termination condition and need for synchronization between threads, enabling parallel execution by multiple threads on individual loop statements. \autoref{fig:while} demonstrates  this strategy yields a consistent speedup, with geometric mean of 1.8 across data sizes and \texttt{gcc} compiler optimization levels.
While this is expectedly slightly lower than speedup of 3.1 obtained on an equivalent \prc|for| loop program---texttt{3mm} is in \autoref{tab:speedup}---this is an encouraging result because \prc|while| loops are the most demanding types of loops to optimize, and generally completely ignored.
Similar examples illustrating that our analysis can apply to \eg loop-carrying dependency loops or loops whose iteration space is not known at compilation time could be similarly crafted.

\begin{figure}
	\begin{minipage}[c]{0.49\linewidth}
		\includegraphics[width=.9\linewidth]{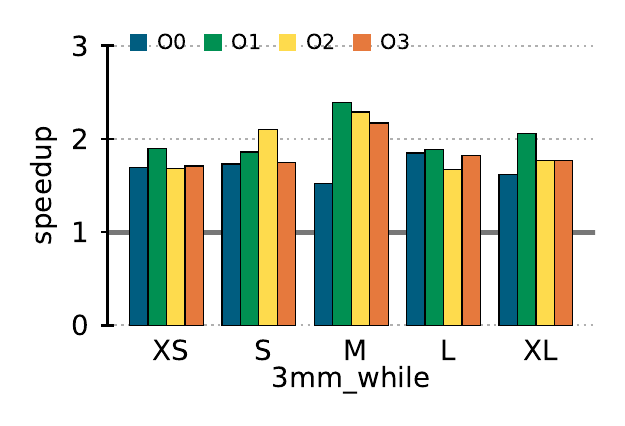}
	\end{minipage}
	\begin{minipage}[c]{0.49\linewidth}
		\begin{center}
			\begin{tabular}{| c || c | c | c | c |} \hline
   Size &  O0   &  O1   &  O2  & O3   \\ \hline \hline
               XS   & \gradient{1.69} & \gradient{1.90} & \gradient{1.68} & \gradient{1.71} \\ \hline
               S    & \gradient{1.73} & \gradient{1.86} & \gradient{2.10} & \gradient{1.75} \\ \hline
               M    & \gradient{1.52} & \gradient{2.39} & \gradient{2.29} & \gradient{2.17} \\ \hline
               L    & \gradient{1.85} & \gradient{1.89} & \gradient{1.67} & \gradient{1.82} \\ \hline
               XL   & \gradient{1.62} & \gradient{2.06} & \gradient{1.77} & \gradient{1.77} \\ \hline
\end{tabular}

		\end{center}
	\end{minipage}
	\caption{Speedup of \texttt{3mm} program implemented using \prc|while| loops.}
	\label{fig:while}
\end{figure} 

\subsection{PolyBench Results}
\label{ssec:results}

\begin{figure}%
\includegraphics[width=1\linewidth]{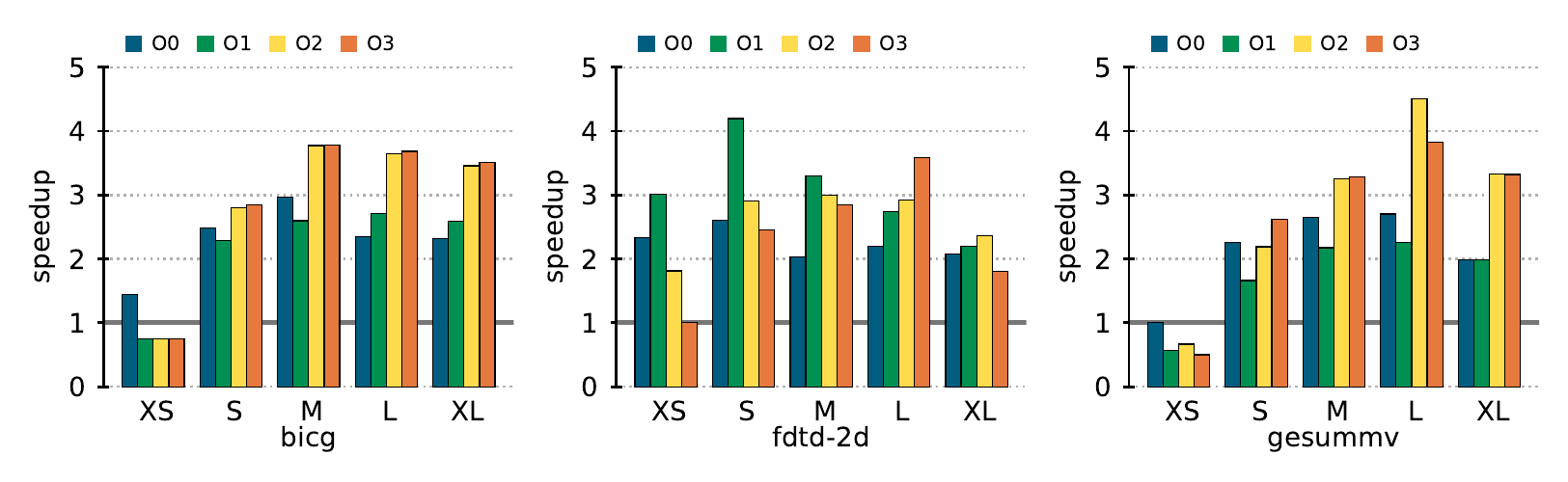}
\caption{Speedup obtained on Polybench/C benchmarks, after loop fission and manual parallelization, on different data sizes and compiler optimization levels.}
\label{fig:speedup}
\end{figure} 

\autoref{fig:speedup} visualizes the speedup obtained on \href{https://github.com/statycc/icc-fission/tree/d684ee4546d9d73d379fa7fb5763c4822945b126/fission_manual}{manually parallelized and transformed programs}.
Included in the figure are only those programs to which loop fission could be applied.
Parallelization directives were applied only to outer loops to reduce parallelization overhead.
Loop chains without flow dependencies were placed inside a \prc|parallel| block and parallelized using \prc|nowait|. Shared variables were marked \prc|private| and \prc|reduction| clauses were used when applicable.

While the performance gains are obvious, a few cases require explanation. Speedup is not always obtained on very small data sizes (\texttt{bicg}, \texttt{gesummv}), because these are very fast programs---less than 0.01 ms in clock time---with low iterations counts, and parallel synchronization produces more overhead to not result in gain.  Nonetheless, any program implemented without parallelism in mind would obtain potential speedup, equalling up to the number of available processors.

\subsubsection{Manual vs. Automatic Parallelization}
\label{ssec:par-techniques}

Finally, we compared three parallelization strategies:
\begin{enumerate*}
	\item applying \texttt{AutoPar-Clava} to the original code,
	\item applying \texttt{AutoPar-Clava} to the code as transformed by our algorithm,
	\item inserting manually the OpenMP directive in the code resulting from our algorithm.
	\end{enumerate*}
Indeed, while \texttt{Autopar-Clava} does not perform loop transformations, it offers a viable alternative for placing the parallelization directives automatically. 
Our results, presented in \autoref{tab:speedup} (in \autoref{app:par-techniques}) makes it evident that all 3 alternative approaches offer comparable results\footnote{Except in the case of \texttt{fdtd-2d}, \href{https://github.com/specs-feup/specs-lara/issues/1}{where the automatic tool was unable to  produce safe parallelization directives}, and no timing result could be obtained.}.
This illustrates that, when applicable, our technique produces results \enquote{as good as} a state-of-the-art parallelization tool, and that it can actually be leveraged in its flow.

\section{Conclusion}
\label{sec:conclusion}

By reasoning about programs at a high-level and in terms of graphs, we obtained an adaptable and correct code transformation, that can be used in a large variety of situations---insensible to the actual tools or programming language provided it is in an imperative style---and that offers notable gain on optimizations that are frequently overlooked.
Particularly, the ability to reason about loops with unknown iteration spaces or loop-carried dependencies is significant, as this property is not supported in the described form by similar existing tools (\autoref{app:sec:comparison}). 
Furthermore, our code transformation is lightweight, automatable, suitable to various forms of implementations, and proven correct. %

From there, multiple perspectives are open: integrating our method in existing tools should not raise difficulties, but will require cost-benefit analysis along, at least, two axes.
The first one is to determine if splitting the loop is performed in the correct type of environment: since \eg parallelizing only the inner-most loop with OpenMP is detrimental to performances~\cite[Chapter 3, Nested]{ppc2022}, taking into account the tool's limitations will be crucial to result in actual gains.
Second, the cost of duplication some commands can void the gain obtained from parallelizing some loops: luckily, our cost analysis~\cite{Aubert2022b} functions on the same imperative language, re-uses some of the tools introduced here, and could help in evaluating the degree of splitting in terms of cost-benefit analysis.
Last but not least, \autoref{thm} would still be valid for this \enquote{conditional} loop-splitting algorithm: further discussion and an extended example including these considerations can be found in the appendix, \autoref{app:sec:cost-benefit}.

\subsubsection*{Acknowledgments}
The authors wish to express their gratitude to \href{https://github.com/joaobispo}{João Bispo} for \href{https://github.com/specs-feup/clava/issues/58}{explaining how to integrate \texttt{AutoPar-Clava}} in \href{https://github.com/statycc/icc-fission}{their benchmark}.

\bibliographystyle{splncs04}
\bibliography{bib/bib.bib}

\clearpage
\appendix

\section{Limitations of Existing Automatic Parallelization Tools}
\label{app:sec:comparison}

We focus here on presenting the types of loop that other \enquote{popular}~\cite{Prema2019} auto-parallelization frameworks for \texttt{C} \emph{cannot} parallelize but that our algorithm could split.
In particular, we do not discuss loops containing function calls that have side effects or control-flow modifiers (such as \prc|break;| or \prc|continue;|): neither our algorithm nor the underlying dependency mechanisms of the discussed tools---to the best of our knowledge---can accommodate those.

Most tools can process only \enquote{canonical loops}, defined \eg in OpenMP's specification~\cite[4.4.1 Canonical Loop Nest Form]{OARB21}: essentially, their structure is of the form \prc|for (init-expr; test-expr; incr-expr) structured-block|, with \prc|incr-expr| being a (single) increment or decrement by a constant or a variable, and \prc|test-expr| being a single comparison between a variable and a variable or a constant.
Additional constraints on loop dependences are sometimes needed, \eg the absence of loop-carried dependency~\cite{enwiki:1080087398} for cetus.

It is always hard to infer the absence of support, but we evaluated the lack of formal discussion or example of \eg \prc|while| loop to be sufficient to determine that the tool could not process \prc|while| loops, unless of course they can trivially be transformed into \prc|for| loops of the required form~\cite[p. 236]{Vitorovi2014}.
We refer to a recent study~\cite[Section 2]{Prema2019} for more detail on those notions and on the limitations of some of the tools discussed below.

It seems further that some tools cannot parallelize loops whose body contains \eg \prc|if| or \prc|switch| statements~\cite[p. 18]{Prema2019}, but we have not investigated this claim further: however, our algorithm can handle \prc|if|---and \prc|switch| too, if it was part of our syntax---present in the body of the loop seamlessly.

\begin{tabular}{| c || c | c | c | }
\hline
Name & \prc|for| loop & \prc|while| loop & \prc|do| \dots \prc|while| loop \\
\hline 
\hline
cetus & In canonical form\textsuperscript{$*$} & \multicolumn{2}{c|}{No}\\
\hline
Par4all  & \multicolumn{3}{c |}{Unknown\textsuperscript{$\dagger$}}\\ 
\hline 
ROSE  & In canonical form\textsuperscript{$\mathsection$} & \multicolumn{2}{c | }{No} \\
\hline
icc & \multicolumn{2}{c |}{Only if countable\textsuperscript{$\ddagger$}} & No\\ 
\hline 
PPL & \multicolumn{3}{c | }{ Unclear\textsuperscript{$\mathparagraph$} }\\
\hline 
Openmp & In canonical loop nest form\textsuperscript{$**$} & \multicolumn{2}{c |}{No} \\
\hline
AutoparClava & \multicolumn{3}{c | }{ Same limitations as openmp} \\
\hline
\end{tabular}
 
\begin{description}
	\item[$*$] 
	\begin{pquotation}{\cite[p.~39]{Dave2009}}
		\textins{cetus} currently handles all canonical loops of the form
		
		\prc|for(i = lb; i < ub; i + = inc)|
	\end{pquotation}
	\begin{pquotation}{\cite[p.~761]{Bae2013}}
		 \textins{a} loop is marked as parallel if no scalar variable carries dependences and all dependence arcs in the graph show non-loop-carried dependences with respect to the loop.
	\end{pquotation}
	\item[$\dagger$] Unfortunately, the project's documentation is \href{http://par4all.github.io/documentation.html\#users-guide}{currently not accessible} and the publications related to this project~\cite{amini2012,amini} do not discuss loop limitations.
	\item[$\mathsection$] 
	Even if the manual boldly claims
	\begin{pquotation}{\cite[p.~123]{Rose}}
		The implementation can successfully optimize arbitrary loop structures, including complex, non-perfect loop nest
	\end{pquotation}
	it later on specifies:
	\begin{pquotation}{\cite[p.~124]{Rose}}
	\textins{Rose} utilizes traditional techniques developed to optimize loop nests in Fortran programs. When optimizing C or C++ applications, this package only recognizes and optimizes a particular for-loop that corresponds to the DO loop construct in Fortran programs. Within the ROSE source-to-source compiler infrastructure, such a loop is defined to have the following formats: 
	
	\prc|for (i = lb; i <= ub; i+ = positiveStep)|
	
	or \prc|for (i = ub; i >= lb; i+ = negativeStep)|
	\end{pquotation} 
		
	\item[$\ddagger$] 
	Loops can be formed with the usual \prc|for| and \prc|while| constructs, provided	the loop iteration is \emph{countable}:
	\begin{pquotation}{\cite[p. 2126]{icc_manual}}
		The loop iterations must be countable; in other words, the number of iterations must be expressed as one of the following: 
		\begin{itemize}
			\item A constant.
			\item A loop invariant term.
			\item A linear function of outermost loop indices. 
		\end{itemize}
		In the case where a loops exit depends on computation, the loops are not countable.
	\end{pquotation}
\item[$\mathparagraph$] The documentation~\cite{tylermsft_parallel_nodate} does not discuss which types of loop are supported clearly, but this tool seems to support only \prc|for| (and \prc|for each|) C++ loops.
\item[$**$] For more detail, refer to OpenMP's documentation~\cite{OARB21}. In short, 
\begin{pquotation}{\cite[Section 4.4.2]{OARB21}}
	The canonical loop nest form allows the iteration count of all associated loops to be computed before executing the outermost loop.
\end{pquotation}
\end{description}

\section{Benchmarking Parallel Versions of PolyBench/C}
\label{app:par-techniques}

\autoref{tab:speedup} presents the gain obtained by parallelizing six different programs from the PolyBench/C suite, for different levels of optimization and data sizes, for the following 4 program categories:

\begin{enumerate}
	\item original: unmodified and sequential programs (used as a base to measure the speedup),
	\item original-autopar: original programs, where OpenMP directives are automatically inserted by \texttt{AutoPar-Clava}~\cite{Arabnejad2020},
	\item fission-manual: programs after loop fission, where OpenMP directives are inserted manually,
	\item fission-autopar: programs after loops fission, where OpenMP directives are automatically inserted by \texttt{AutoPar-Clava}.
\end{enumerate}

\begin{table}
  \centering
		\noindent
\resizebox{1\textwidth}{!}{
\renewcommand{\arraystretch}{1.2}
\begin{tabular}{| p{1.3cm} >{\centering\arraybackslash}p{.7cm} || >{\centering\arraybackslash}p{.7cm} | >{\centering\arraybackslash}p{.7cm} | >{\centering\arraybackslash}p{.7cm} | >{\centering\arraybackslash}p{.7cm} | >{\centering\arraybackslash}p{.7cm} | >{\centering\arraybackslash}p{.7cm} | >{\centering\arraybackslash}p{.7cm} | >{\centering\arraybackslash}p{.7cm} | >{\centering\arraybackslash}p{.7cm} | >{\centering\arraybackslash}p{.7cm} | >{\centering\arraybackslash}p{.7cm} | >{\centering\arraybackslash}p{.7cm} | } \hline
    \multicolumn{2}{|l||}{Benchmark}  & \multicolumn{3}{c|}{O0} & \multicolumn{3}{c|}{O1} & \multicolumn{3}{c|}{O2}  & \multicolumn{3}{c|}{O3}  \\ \hline
     Name  &   Data \newline size   & orig. \newline auto & fiss. \newline auto & fiss. \newline man & orig. \newline auto & fiss. \newline auto & fiss. \newline man & orig. \newline auto & fiss. \newline auto & fiss. \newline man &  orig. \newline auto & fiss. \newline auto & fiss. \newline man \\  \hline\hline
    3mm     & XS   & \gradient{2.28} & \gradient{2.29} & \gradient{2.31} & \gradient{2.46} & \gradient{2.50} & \gradient{2.75} & \gradient{1.73} & \gradient{1.73} & \gradient{1.98} & \gradient{1.81} & \gradient{1.96} & \gradient{2.04} \\ \hline
            & S    & \gradient{2.76} & \gradient{2.79} & \gradient{2.58} & \gradient{3.91} & \gradient{3.92} & \gradient{3.95} & \gradient{4.19} & \gradient{4.19} & \gradient{4.21} & \gradient{3.42} & \gradient{3.41} & \gradient{3.47} \\ \hline
            & M    & \gradient{2.24} & \gradient{2.23} & \gradient{2.23} & \gradient{3.63} & \gradient{3.61} & \gradient{3.63} & \gradient{3.42} & \gradient{3.40} & \gradient{3.43} & \gradient{3.45} & \gradient{3.44} & \gradient{3.47} \\ \hline
            & L    & \gradient{3.56} & \gradient{3.49} & \gradient{3.48} & \gradient{3.97} & \gradient{3.85} & \gradient{3.91} & \gradient{4.16} & \gradient{4.05} & \gradient{3.98} & \gradient{4.40} & \gradient{4.45} & \gradient{4.44} \\ \hline
            & XL   & \gradient{2.35} & \gradient{2.25} & \gradient{2.31} & \gradient{3.96} & \gradient{4.06} & \gradient{3.76} & \gradient{2.92} & \gradient{2.87} & \gradient{2.88} & \gradient{2.76} & \gradient{2.78} & \gradient{2.85} \\ \hline
    bicg    & XS   & \gradient{0.53} & \gradient{0.45} & \gradient{1.44} & \gradient{0.23} & \gradient{0.23} & \gradient{0.75} & \gradient{0.23} & \gradient{0.20} & \gradient{0.75} & \gradient{0.28} & \gradient{0.26} & \gradient{0.75} \\ \hline
            & S    & \gradient{2.08} & \gradient{1.82} & \gradient{2.48} & \gradient{1.68} & \gradient{1.37} & \gradient{2.28} & \gradient{1.57} & \gradient{1.45} & \gradient{2.80} & \gradient{1.73} & \gradient{1.69} & \gradient{2.85} \\ \hline
            & M    & \gradient{3.30} & \gradient{2.84} & \gradient{2.96} & \gradient{3.56} & \gradient{2.43} & \gradient{2.60} & \gradient{4.10} & \gradient{3.37} & \gradient{3.77} & \gradient{4.20} & \gradient{3.48} & \gradient{3.78} \\ \hline
            & L    & \gradient{2.74} & \gradient{2.34} & \gradient{2.35} & \gradient{3.96} & \gradient{2.63} & \gradient{2.71} & \gradient{4.54} & \gradient{3.63} & \gradient{3.64} & \gradient{4.56} & \gradient{3.68} & \gradient{3.68} \\ \hline
            & XL   & \gradient{2.71} & \gradient{2.30} & \gradient{2.32} & \gradient{3.77} & \gradient{2.60} & \gradient{2.59} & \gradient{4.27} & \gradient{3.46} & \gradient{3.46} & \gradient{4.30} & \gradient{3.50} & \gradient{3.50} \\ \hline
    deriche & XS   & \gradient{2.03} & \gradient{2.08} & \gradient{1.63} & \gradient{2.24} & \gradient{2.19} & \gradient{2.26} & \gradient{2.36} & \gradient{2.54} & \gradient{2.29} & \gradient{2.20} & \gradient{2.22} & \gradient{2.35} \\ \hline
            & S    & \gradient{2.33} & \gradient{2.28} & \gradient{1.76} & \gradient{2.20} & \gradient{2.33} & \gradient{1.94} & \gradient{2.42} & \gradient{2.53} & \gradient{2.01} & \gradient{2.34} & \gradient{2.29} & \gradient{2.08} \\ \hline
            & M    & \gradient{2.73} & \gradient{2.74} & \gradient{1.96} & \gradient{3.03} & \gradient{3.05} & \gradient{2.44} & \gradient{3.10} & \gradient{3.20} & \gradient{2.38} & \gradient{2.95} & \gradient{3.03} & \gradient{2.56} \\ \hline
            & L    & \gradient{1.73} & \gradient{1.73} & \gradient{1.39} & \gradient{1.70} & \gradient{1.55} & \gradient{1.72} & \gradient{1.75} & \gradient{1.75} & \gradient{1.65} & \gradient{1.72} & \gradient{1.71} & \gradient{1.86} \\ \hline
            & XL   & \gradient{0.95} & \gradient{0.95} & \gradient{0.89} & \gradient{0.84} & \gradient{0.84} & \gradient{0.84} & \gradient{0.86} & \gradient{0.86} & \gradient{0.84} & \gradient{0.86} & \gradient{0.86} & \gradient{0.87} \\ \hline
    fdtd-2d & XS   & \gradient{2.17} & \- & \gradient{2.33} & \gradient{2.13} & \- & \gradient{3.01} & \gradient{1.42} & \- & \gradient{1.81} & \gradient{0.77} & \- & \gradient{1.00} \\ \hline
            & S    & \gradient{2.60} & \- & \gradient{2.60} & \gradient{3.47} & \- & \gradient{4.20} & \gradient{2.28} & \- & \gradient{2.90} & \gradient{1.59} & \- & \gradient{2.45} \\ \hline
            & M    & \gradient{1.93} & \- & \gradient{2.03} & \gradient{1.39} & \- & \gradient{3.29} & \gradient{1.27} & \- & \gradient{3.00} & \gradient{0.80} & \- & \gradient{2.85} \\ \hline
            & L    & \- & \- & \gradient{2.20} & \- & \- & \gradient{2.74} & \- & \- & \gradient{2.92} & \- & \- & \gradient{3.58} \\ \hline
            & XL   & \- & \- & \gradient{2.07} & \- & \- & \gradient{2.20} & \- & \- & \gradient{2.36} & \- & \- & \gradient{1.80} \\ \hline
    gesummv & XS   & \gradient{1.50} & \gradient{0.90} & \gradient{1.00} & \gradient{1.00} & \gradient{0.57} & \gradient{0.57} & \gradient{1.00} & \gradient{0.57} & \gradient{0.67} & \gradient{0.75} & \gradient{0.43} & \gradient{0.50} \\ \hline
            & S    & \gradient{2.71} & \gradient{2.33} & \gradient{2.26} & \gradient{2.58} & \gradient{1.55} & \gradient{1.66} & \gradient{2.36} & \gradient{2.04} & \gradient{2.19} & \gradient{2.82} & \gradient{2.62} & \gradient{2.62} \\ \hline
            & M    & \gradient{3.07} & \gradient{2.80} & \gradient{2.65} & \gradient{3.27} & \gradient{2.17} & \gradient{2.17} & \gradient{3.19} & \gradient{3.25} & \gradient{3.25} & \gradient{3.29} & \gradient{3.29} & \gradient{3.28} \\ \hline
            & L    & \gradient{3.08} & \gradient{2.85} & \gradient{2.70} & \gradient{3.48} & \gradient{1.95} & \gradient{2.26} & \gradient{4.23} & \gradient{4.51} & \gradient{4.51} & \gradient{3.56} & \gradient{3.78} & \gradient{3.83} \\ \hline
            & XL   & \gradient{2.27} & \gradient{2.10} & \gradient{1.99} & \gradient{3.16} & \gradient{2.02} & \gradient{1.99} & \gradient{3.16} & \gradient{3.32} & \gradient{3.32} & \gradient{3.15} & \gradient{3.31} & \gradient{3.32} \\ \hline
    mvt     & XS   & \gradient{1.73} & \gradient{1.73} & \gradient{2.17} & \gradient{1.43} & \gradient{1.43} & \gradient{2.00} & \gradient{0.83} & \gradient{0.83} & \gradient{1.25} & \gradient{0.83} & \gradient{1.00} & \gradient{1.25} \\ \hline
            & S    & \gradient{2.73} & \gradient{2.75} & \gradient{2.85} & \gradient{3.64} & \gradient{3.56} & \gradient{3.56} & \gradient{2.54} & \gradient{2.49} & \gradient{2.85} & \gradient{2.51} & \gradient{2.38} & \gradient{2.71} \\ \hline
            & M    & \gradient{3.10} & \gradient{3.06} & \gradient{3.44} & \gradient{4.33} & \gradient{4.33} & \gradient{4.38} & \gradient{3.14} & \gradient{3.14} & \gradient{3.16} & \gradient{3.15} & \gradient{3.19} & \gradient{3.20} \\ \hline
            & L    & \gradient{2.28} & \gradient{2.30} & \gradient{2.22} & \gradient{3.53} & \gradient{2.97} & \gradient{3.49} & \gradient{2.53} & \gradient{2.50} & \gradient{2.51} & \gradient{2.51} & \gradient{2.57} & \gradient{2.50} \\ \hline
            & XL   & \gradient{1.41} & \gradient{1.35} & \gradient{1.32} & \gradient{2.28} & \gradient{1.99} & \gradient{2.38} & \gradient{1.33} & \gradient{1.53} & \gradient{1.52} & \gradient{1.55} & \gradient{1.48} & \gradient{1.53} \\ \hline
\end{tabular}
}

\caption{Comparing speedup between original sequential programs and transformed parallel programs, for various data sizes and compiler optimization levels.}\label{tab:speedup}
\end{table}

\section{Cost-Benefit Considerations for Loop Fission}
\label{app:sec:cost-benefit}

For the example transformations presented in \autoref{ssec:algo}, it may be that duplicating the command $y_2 \gets y_2 \times y_1$ degrades the speed gain of the parallelization to the point that it is not worth splitting the loop.
Based on a more precise dependency analysis, such as the \textsc{mwp}-analysis~\cite{Aubert2022b}, which produces a weighted dependency graph whose weights \emph{provide quantitative information} on the dependency (linear, weak polynomial, polynomial, super-polynomial). The splitting algorithm will be adapted to use this information in order to allow for adaptive splittings. In the above case, one may then obtain the cover and corresponding algorithm presented in \autoref{ex:optimized-mwp}.
Note that \autoref{thm} would still apply to this \enquote{conditional} loop-splitting algorithm.

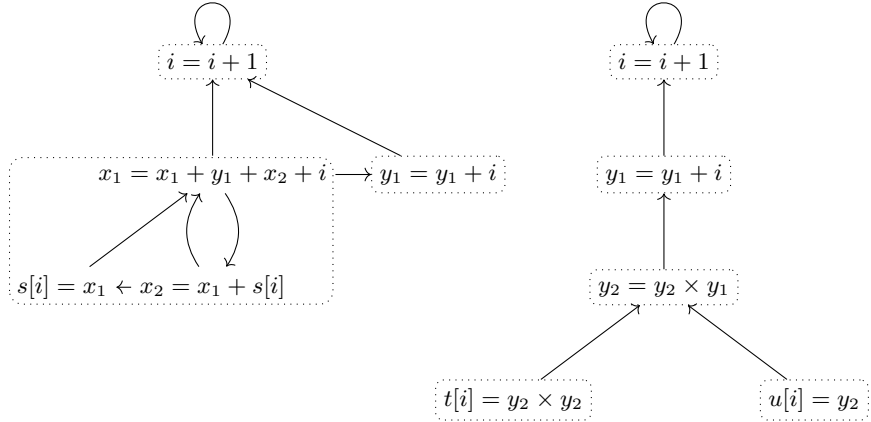
\begin{figure}
	\begin{algorithmic}
		\State $j \gets \textbf{pow}(10, 6)$
		\State $x_1 \gets 1$
		\State $i \gets 1$
		\State \prc|parallel|  \Comment{A private copy of \(i\) and \(y_1\) needs to be given to each loop below.}
		\vspace{-1em}
		\begin{multicols}{2}
			\noindent $\overbrace{\hspace{10em}}$
			\While{$i \neq j$}
			\State $x_1 \gets x_1 + y_1 + x_2 + i$
			\State $y_1 \gets y_1+i$
			\State $s[i] \gets x_1$
			\State $x_2 \gets x_1+s[i]$
			\State $i \gets i+1$
			\EndWhile
			\vspace{-.7em}
			
			\noindent $\underbrace{\hspace{10em}}$
			
			\noindent $\overbrace{\hspace{10em}}$
			\While{$i \neq j$}%
			\State $y_1 \gets y_1+i$
			\State $y_2 \gets y_2\times y_1$
			\State $u[i] \gets y_2$
			\State $t[i] \gets y_2 \times y_2$  
			\State $i \gets i+1$
			\EndWhile
			\vspace{-.7em}
			
			\noindent $\underbrace{\hspace{10em}}$
		\end{multicols}
		\Use{$x_1$}
		\Comment{The copies of \(i\) and \(y_1\) can be destroyed, all have the same values.}
	\end{algorithmic}

	\begin{tikzpicture}[x=2cm,y=1.5cm]
	\node[rectangle,dotted,rounded corners=3] (1) at (0,0) {\footnotesize{$x_1 \gets x_1 + y_1 + x_2 + i$}};
	\node[draw,rectangle,dotted,rounded corners=3] (2) at (1.5,0) {\footnotesize{$y_1 \gets y_1+i$}};
	\node[draw,rectangle,dotted,rounded corners=3] (2bis) at (3,0) {\footnotesize{$y_1 \gets y_1+i$}};
	\node[draw,rectangle,dotted,rounded corners=3] (3) at (3,-1) {\footnotesize{$y_2 \gets y_2\times y_1$}};
	\node[rectangle,dotted,rounded corners=3] (4) at (-1,-1) {\footnotesize{$s[i] \gets x_1$}};
	\node[rectangle,dotted,rounded corners=3] (5) at (0,-1) {\footnotesize{$x_2 \gets x_1+s[i]$}};
	\node[draw,rectangle,dotted,rounded corners=3] (6) at (4,-2) {\footnotesize{$u[i] \gets y_2$}};
	\node[draw,rectangle,dotted,rounded corners=3] (7) at (2,-2) {\footnotesize{$t[i] \gets y_2 \times y_2$}};
	\node[draw,rectangle,dotted,rounded corners=3] (8) at (0,1) {\footnotesize{$i \gets i+1$}};
	\node[draw,rectangle,dotted,rounded corners=3] (8bis) at (3,1) {\footnotesize{$i \gets i+1$}};
	
	\draw[->] (7)--(3);
	\draw[->] (1)--(2);
	\draw[->] (3)--(2bis);
	\draw[->] (6)--(3);
	\draw[->] (4)--(1);
	\draw[->] (5) .. controls (-0.2,-0.6) and (-0.2,-0.4) .. (1);
	\draw[->] (5)--(4);
	\draw[->] (1) .. controls (0.2,-0.4) and (0.2,-0.6) .. (5);
	\draw[->] (1)--(8);
	\draw[->] (2)--(8);
	\draw[->] (2bis)--(8bis);
	\draw[->] (8) to [out=60,in=120,looseness=8] (8);
	\draw[->] (8bis) to [out=60,in=120,looseness=8] (8bis);
	
	\draw[dotted,rounded corners=5] (-1.35,-1.15) rectangle (0.8,0.15);
	\end{tikzpicture}
	\label{ex:optimized-mwp}
	\caption{Optimized Example}
\end{figure}

\end{document}